\documentclass[a4paper,pra,twocolumn,showpacs,groupedaddress,superscriptaddress]{revtex4-1}
\usepackage{epsfig}
\usepackage{amsbsy,latexsym}
\usepackage{amsmath}
\usepackage{amsthm}
\usepackage{amssymb, mathrsfs}
\usepackage[mathscr]{eucal}
\usepackage{hyperref}
\usepackage{color}
\usepackage{subfigure}
\usepackage[normalem]{ulem}
\usepackage{cancel}

\def\<{\langle}
\def\>{\rangle}

\def\kk{\rangle\!\rangle}
\def\bb{\langle\!\langle}

\def\sZ{B}

\def\ket#1{|#1\>}
\def\bra#1{\<#1|}

\newcommand{\cI}{\mathcal{I}}
\newcommand{\cB}{\mathcal{B}}
\newcommand{\cE}{\mathcal{E}}
\newcommand{\cF}{\mathcal{F}}
\newcommand{\cU}{\mathcal{U}}
\newcommand{\cS}{\mathcal{S}}
\newcommand{\cH}{\mathcal{H}}

\newcommand{\scalar}[2]{\langle #1 | #2 \rangle}
\newcommand{\ketbra}[2]{| #1 \rangle \langle #2 |}

\newcommand{\tr}[1]{\mathrm{tr}[#1]}

\newtheorem{lemma}{Lemma}
\newtheorem{prop}{Proposition}
\newtheorem{theorem}{Theorem}
\newtheorem{corollary}{Corollary}

\newcommand{\ii}{\mathrm{i}}

\begin{document}
\title{Exploring boundaries of quantum convex structures: special role of unitary processes}
\author{Zbigniew Pucha\l{}a}
\affiliation{Institute of Theoretical and Applied Informatics, Polish Academy of Sciences, ulica Ba\l{}tycka 5, 44-100 Gliwice, Poland}
\affiliation{Institute of Physics, Jagiellonian University, ulica prof. Stanis\l{}awa \L{}ojasiewicza 11, 30-348 Krak\'o{}w, Poland}

\author{Anna Jen\v cov\'a}
\affiliation{Institute of Mathematics,~Slovak Academy of Sciences, \v Stef\'anikova,~84511 Bratislava,~Slovakia}
\author{Michal Sedl\'ak}
\affiliation{Department of Optics, Palack\'y University, 17. listopadu 1192/12,~77146 Olomouc, Czech Republic}
\affiliation{Institute of Physics,~Slovak Academy of Sciences,~D\'ubravsk\'a cesta 9,~84511 Bratislava,~Slovakia}
\author{M\'ario Ziman}
\affiliation{Institute of Physics,~Slovak Academy of Sciences,~D\'ubravsk\'a cesta 9,~84511 Bratislava,~Slovakia}
\affiliation{Faculty of Informatics,~Masaryk University,~Botanick\'a 68a,~60200 Brno, Czech Republic}

\begin{abstract}
We address the question of finding the most effective convex decompositions
into boundary elements (so-called boundariness) for sets of quantum states,
observables and channels. First we show that in general convex sets the
boundariness essentially coincides with the question of the most
distinguishable element, thus, providing an operational meaning for this
concept. Unexpectedly, we discovered that for any interior point of the set of
channels the optimal decomposition necessarily contains a unitary channel. In
other words, for any given channel the best distinguishable one is some unitary
channel. Further, we prove that boundariness is sub-multiplicative under
composition of systems and explicitly evaluate its maximal value that is
attained only for the most mixed elements of the considered convex structures.
\end{abstract}
\pacs{3.67.-a}
\maketitle

%%%%%%%%%%%%%%%%%%%%%%%%%%%%%%%%%%%%%%%%%%%%%%%%%%%%%%%%%%%
\section{Introduction}
%%%%%%%%%%%%%%%%%%%%%%%%%%%%%%%%%%%%%%%%%%%%%%%%%%%%%%%%%%%
Convexity, rooted in the very concept of probability,
is one of unavoidable mathematical features of our description
of physical systems. Operationally, it originates in our
ability to switch randomly between different physical devices
of the same type. As a result, all elementary quantum structures
and most of the quantum properties are "dressed in convex clothes".
For example, the sets of states, observables and processes are
all convex, and it is of foundational interest to understand the
similarities and identify the differences of their convex structures.

For any convex set, we may introduce the concept of an interior point in a
natural way as a point that can be connected to any other point by a line
segment containing it in its interior. We will use this concept to define
mixedness and boundariness as measures evaluating how much the element is not
extremal, or how much the element is not a boundary point, respectively. More
precisely,  mixedness  will be determined via the highest weight occurring in
decompositions into extremal points and boundariness  will be determined via
the highest weight occurring in decompositions into boundary points. In both
cases, these numbers tell us how much randomness is needed to create the given
element. Since we focus on sets of quantum devices related to finite dimensional
Hilbert spaces, we will work in finite dimensional setting, but note that
similar definitions can be introduced also in infinite dimensions, although some
of the facts used below are no longer true.

If the given convex set is also compact, it can be viewed as a base of a closed
pointed convex cone and we may consider the corresponding base norm in the
generated vector space (see e.g. \cite{rockafellar}). Note that the related
distance between points of the base can be determined solely from the convex
structure of the base (see for instance recent works
\cite{reeb_etal2011,errka}). As it is well known for
quantum states \cite{Helstrom, Holevo} and as has been  recently proved for
other quantum devices \cite{jencova2013}, this distance is closely related to
the minimum error discrimination problem.

It was proved in Ref. \cite{haapasalo2014} that for the sets of quantum states
and observables, boundariness and the base norm distance are closely related.
More precisely, the largest distance of a given interior point $y$ from another
point of the base is given in terms of boundariness of $y$. In the present
paper, we show that this is true for any base of the positive cone in a finite
dimensional ordered vector space. In particular, for sets of quantum devices,
this property singles out a subset of extremal elements that are best
distinguishable from interior points. Exploiting these results, we will point
out an interesting difference between the convex sets of states and channels,
and also provide an unexpected operational characterization of unitary channels.

This paper is organized as follows. In Section II we will provide readers with
basic elements of convex analysis and quantum theory relevant for the rest
of the paper. The concept of boundariness will be introduced in Section III,
where various equivalent definitions will be stated and also its operational
meaning will be discussed. In Section IV we will investigate the boundariness
for the case of quantum channels. In particular, we will prove a conjecture
stated in Ref.~\cite{haapasalo2014}. In Section V we will address the question of
boundariness for composition of systems and Section VI is devoted to
identification of elements for which boundariness achieves its maximal
value. Last Section VII summarizes our results.

%%%%%%%%%%%%%%%%%%%%%%%%%%%%%%%%%%%%%%%%%%%%%%%%%%%%%%%%%%%
\section{Quantum convex cone structures}
%%%%%%%%%%%%%%%%%%%%%%%%%%%%%%%%%%%%%%%%%%%%%%%%%%%%%%%%%%%
Suppose $V$ is a real finite-dimensional vector space  and $C\subset V$ is a
closed convex cone. We assume that $C$ is \emph{pointed}, i.e.
$C\cap -C=\{0\}$, and \emph{generating}, i.e. $V=C-C$.
Then $(V,C)$ becomes a partially ordered vector
space, with $C$ the cone of positive elements. Let $V^*$ be the dual space with
duality $\<\cdot,\cdot\>$, then we may introduce a partial order in $V^*$ as
well, with the dual cone of positive functionals $C^*=\{f\in V^*, \<f,z\>\ge 0,\
\forall z\in C\}$. Note that $C^*$ is again pointed and generating, and
$C^{**}=C$.

Interior points $z\in int(C)$ of the cone $C$ are characterized by the property
that for each $v\in V$ there is some $t>0$ such that $tz-v\in C$, that is, the
interior points of $C$ are precisely the \emph{order units} in $(V,C)$.
Alternatively, the following lemma gives a well known characterization of
boundary points of $C$ as elements contained in some supporting hyperplane of
$C$, see Ref. \cite[Section 11]{rockafellar} for more details.

\begin{lemma}
An element $z\in C$ is a boundary point, $z\in \partial C$,
if and only if there exists
a nonzero element $f\in C^*$ such that $\<f,z\>=0$. Clearly,
then also $f\in \partial C^*$.
\end{lemma}

A \emph{base} of $C$ is a compact convex subset $B\subset C$ such that for every
nonzero $z\in C$, there is a unique constant $t>0$ and an element $b\in B$ such
that $z=tb$. The \emph{relative interior} $ri(B)$ is defined as the interior of
$B$ with respect to the relative topology in the smallest affine subspace
containing $B$. Note that we have $ri(B)=B\cap int(C)$, so that the boundary
points $z\in \partial B=B\setminus ri(B)$ can be characterized as in the
previous lemma.

There is a one-to-one correspondence between bases $B\subset C$ and order units
in the dual space  $e\in int(C^*)$, such that $B=\{ z\in C, \<e,z\>=1\}$ is a
base of $C$ if and only if $e$ is an order unit. The order unit $e$ determines
the \emph{order unit norm} in $(V^*,C^*)$ as
\[
\|f\|_e=\inf\{\lambda>0, \lambda e\pm f\in C^*\},\quad f\in V^*.
\]
Its dual is the \emph{base norm} $\|\cdot\|_B$ in $(V,C)$. In particular, we
obtain the following expression for the corresponding distance of elements of
$B$:
\begin{align}\label{eq:base}
\|x-y\|_B=2\sup_{g,e-g\in C^*}\<g,x-y\>,\qquad x,y\in B
\end{align}

We will now describe the basic convex sets (see Ref.\cite{heinosaari12}) of
quantum states, channels and measurements (observables). Let us stress that each
of these sets is a compact convex subset in a finite dimensional vector space
and as such forms a base of the positive cone of some partially ordered vector
space, so that these sets fit into the framework introduced above.

Let us denote by $\cH_d$ the $d$-dimensional Hilbert space associated
with the studied physical system. Then $\cS(\cH_d)$ stands for the set
of all density operators (positive linear operators of unit trace)
representing the set of quantum states.

Observables are identified with positive-operator valued measures
(POVMs) being determined by a collection of effects $E_1,\dots,E_m$
($O\leq E_j\leq I$) normalized as $\sum_j E_j=I$. Each effect $E_j$
defines a different measurement outcome. In particular, if the system
is prepared in a state $\varrho$, then $p_j=\tr{\varrho E_j}$ is the
probability of the registration of the $j$th outcome.

Quantum channels are modeled by completely positive trace-preserving linear
maps, i.e. by transformations $\varrho\mapsto \sum_l A_l\varrho A_l^\dagger$ for
any collection of operators $\{A_l\}_l$ satisfying the normalization
$\sum_l A_l^\dagger A_l=I$. Define the one-dimensional projection operator
$\Psi_+=\frac{1}{d}\sum_{j,k}\ket{jj}\bra{kk}$ on $\cH_d\otimes\cH_d$, where the
vectors $\ket{j}$ form a complete orthonormal basis on $\cH_d$. Due to
Choi-Jamiolkowski isomorphism \cite{Jamiolkowski72,Choi75}, the set of quantum
channels of a finite-dimensional quantum system is mathematically closely
related to the set of density operators (states) of a composite system. In
particular, a channel $\cE$ is associated with a density operator
\[
J_\cE=(\cE\otimes\cI)[\Psi_+]\in\cS(\cH_d\otimes\cH_d)
\]
and the normalization condition ${\rm tr}_1 J_\cE=\frac{1}{d}I$ is the only
difference between the mathematical representations of states and channels. In
other words,  only a special (convex) subset of density operators on
$\cH_d\otimes\cH_d$ can be identified with quantum channels on $d$-dimensional
quantum systems.

\section{Boundariness}

For any element of a compact convex subset $\sZ\subset V$ with boundary
$\partial \sZ$ and a set of extremal elements $ext(\sZ)$ we may introduce the
concepts of \emph{mixedness} and \emph{boundariness} evaluating the "distance"
of the element from extremal and boundary points, respectively. For any convex
decomposition $y=\sum_j \pi_j x_j$, where $0\leq\pi_j\leq 1$ and
$\sum_j\pi_j=1$, we define its maximal weight $w_y(\{\pi_j,x_j\}_j)=\max_j
\pi_j$. Using this quantity, we may express the mixedness of $y\in \sZ$ as
follows
$$
m(y)=1-\sup_{x_j\in ext(\sZ)} w_y(\{\pi_j,x_j\}_j)\,,
$$
where supremum is taken over all convex decompositions of $y$ into extremal
elements. In a similar way we may define the boundariness \cite{haapasalo2014}
of $y$ as
\begin{align}
\label{eq:defbyn}
b(y)=1-\sup_{x_j\in\partial \sZ} w_y(\{\pi_j,x_j\}_j)\,,
\end{align}
where supremum is taken over all decompositions into boundary elements. By
definition $m(y)\geq b(y)$, since the convex decompositions in (\ref{eq:defbyn})
are less restrictive.

\begin{figure}[t]
\centering
\subfigure[ ]{%
\includegraphics[width=0.48\linewidth]{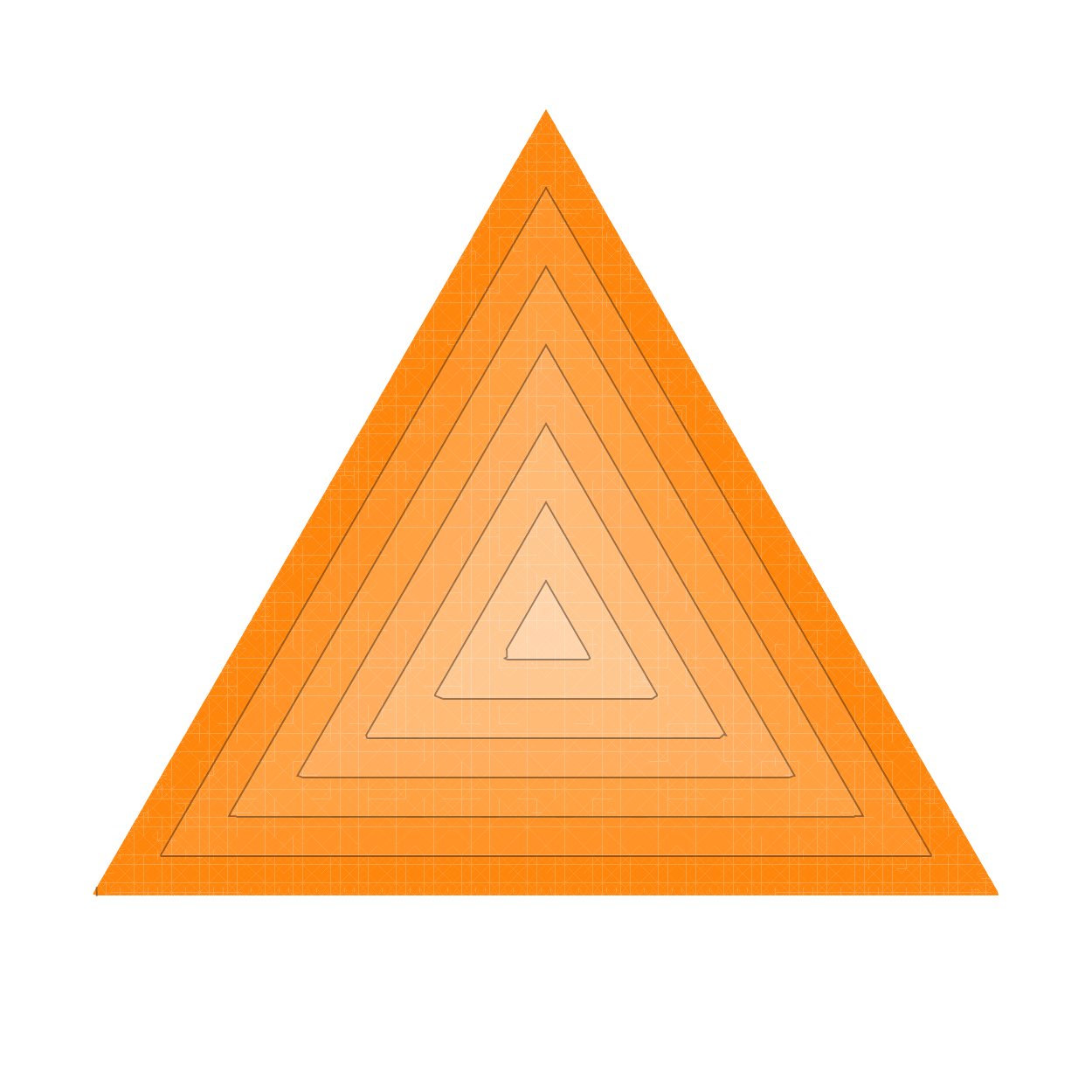}
\label{fig:subfigure1}}
\subfigure[ ]{%
\includegraphics[width=0.48\linewidth]{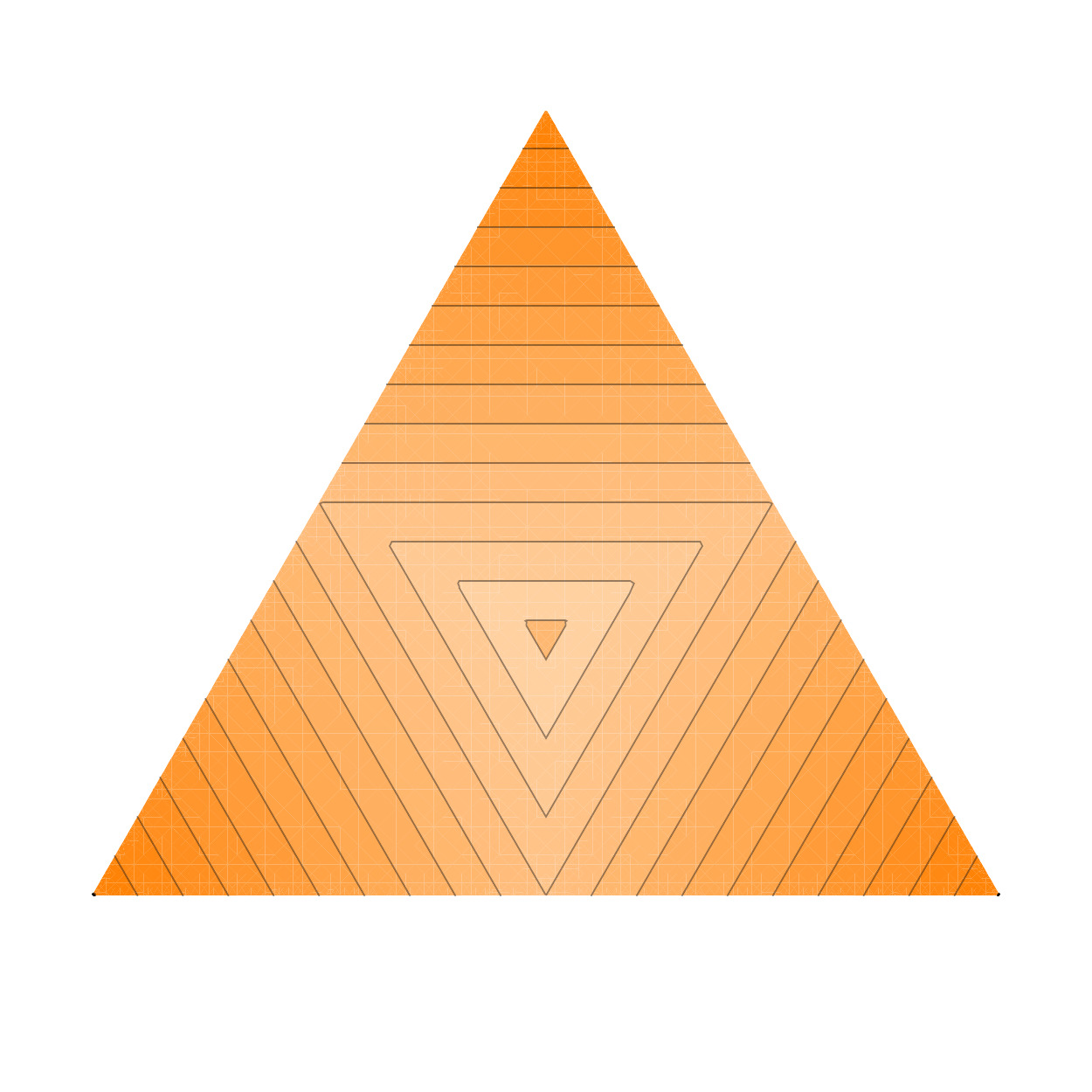}
\label{fig:subfigure2}}\\
\subfigure[ ]{%
\includegraphics[width=0.48\linewidth]{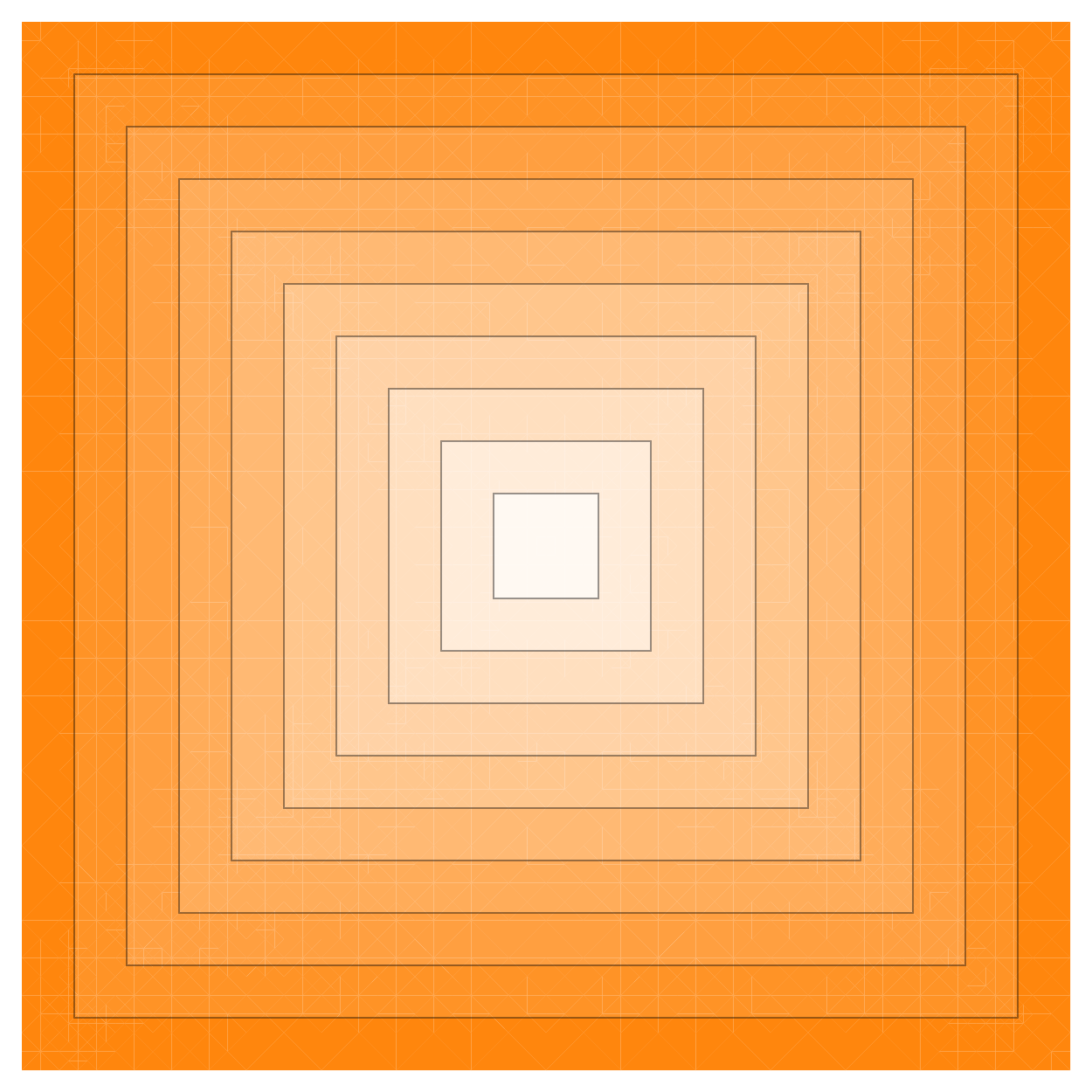}
\label{fig:subfigure3}}
\subfigure[ ]{%
\includegraphics[width=0.48\linewidth]{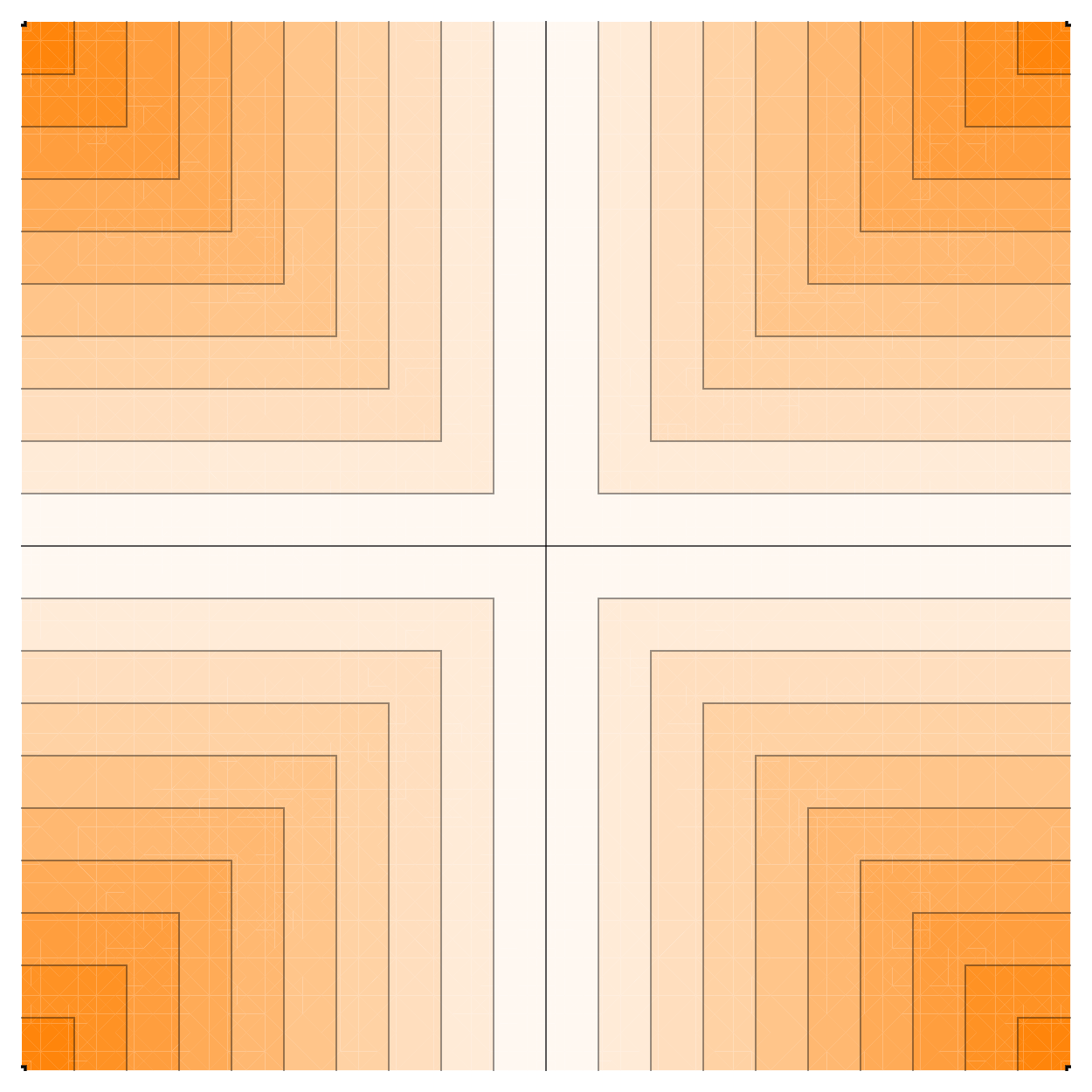}
\label{fig:subfigure4}}
\caption{(Color online) Illustration of boundariness (subfigure a,c) and
mixedness (subfigures b,d) for simple convex sets. }
\label{fig:numresults}
\end{figure}

Let us prove that the above formula is equivalent to the original definition
\cite{haapasalo2014} of boundariness. We recall that for any element $y\in \sZ$,
the {\it weight function} $t_y:\sZ \to [0,1]$ assigns for every $x\in \sZ$ the
supremum of possible weights of the point $x$ in convex decompositions of $y$,
i.e.
$$
t_y(x)=\sup\Big\{0\leq t < 1\,\Big|\,z=\frac{y-tx}{1-t}\in \sZ\Big\}\,.
$$
Thanks to  compactness of $\sZ$, the supremum is really attained and there
exists some $z\in\sZ$ such that $y=t x+(1-t)z$, where $t=t_y(x)$. Note that we
must have $z\in \partial \sZ$ and, in fact, for an interior point $y$,
$t=t_y(x)$ is equivalent to $z\in \partial \sZ$. Let us consider a convex
decomposition $y=\sum_j \pi_j x_j$, $x_j \in \partial \sZ$ and denote by $k$ the
index for which $\pi_k=\max_j \pi_j\neq 1$ (the case $\max_j \pi_j=1$ is trivial
and $b(y)=0$ in both definitions). If we define $\overline{x}_k=\sum_{j\neq k}
\frac{\pi_j}{1-\pi_k}x_j$ then $y=\pi_k x_k + (1-\pi_k)\overline{x}_k$, where
$\overline{x}_k\in \sZ$.  Either $\overline{x}_k \in \partial\sZ$ and we managed
to rewrite $y$ as a two term convex combination of elements from boundary or
$\overline{x}_k \in \sZ\setminus\partial\sZ$, which implies $\pi_k<t_y(x_k)$ and
there exists $w\in\partial\sZ$ such that a better two term decomposition $y=t
x_k+(1-t)w$ with $t>\pi_k$ exists. This shows that definition
(\ref{eq:defbyn}) is equivalent to
\begin{align}
b(y)&=1-\sup_{x,z\in \partial \sZ}\{s|y=(1-s)x+s z\} \nonumber \\
&=\inf_{x\in \partial \sZ}t_y(x)\,.  \nonumber
\end{align}
Finally, we obtain the original definition \cite{haapasalo2014}
\begin{align}
\label{eq:defbdorig}
b(y)=\inf_{x\in \sZ}t_y(x),
\end{align}
because the infimum is always determined by elements $x\in ext(\sZ)$ as we
discussed in Ref.~\cite[Proposition 1]{haapasalo2014}.

Having established the cone picture of quantum structures, it is useful to see
how boundariness can be defined using this language.
\begin{lemma}\label{lemma:bdd_dual}
Let $f\in C^*$. If $\|f\|_e=1$, then $e-f\in \partial C^*$.
\end{lemma}
\begin{proof} Suppose $\|f\|_e=1$, then $e-f\in C^*$. If $e-f\in int(C^*)$, then there is some $t>0$ such that
$e-f\pm tf\in C^*$. But then $(1+t)^{-1}e-f\in C^*$, so that $\|f\|_e\le (1+t)^{-1}<1$.
\end{proof}

We now find an equivalent expression for boundariness.
\begin{prop}
\label{prop:evalbd1}
$b(y)=\min\{ \<f,y\>,\ f\in C^*, \|f\|_e=1\}$.
\end{prop}

\begin{proof}
Let us denote the minimum on the right hand side by $\tilde b(y)$. Let $x\in
\sZ$ and $y=tx+(1-t)z$, with  $t=t_y(x)$. Then $z\in \partial \sZ$, so that
there is some nonzero $f\in C^*$ such that $\<f,z\>=0$. Put $\tilde f=
\|f\|_e^{-1}f$, then $\tilde f\in C^*$, $\|\tilde f\|_e=1$ and we have
\[
\tilde b(y)\le \<\tilde f,y\>=t_y(x)\<\tilde f,x\>\le t_y(x).
\]
Since this holds for all $x\in \sZ$, we obtain $\tilde b(y)\le b(y)$.

For the converse, let $f\in C^*$, $\|f\|_e=1$, then $e-f\in \partial C^*$. Hence
there is some element  $x\in \sZ$, such that $\<e-f,x\>=0$. Let $s=t_y(x)$, then
$y=sx+(1-s)z$ for some $z\in \partial \sZ$. We have
\[
\<f,y\>=1-\<e-f,y\>=1-(1-s)\<e-f,z\>\ge s=t_y(x)\ge b(y),
\]
hence $\tilde b(y)\ge b(y)$.
\end{proof}

Let $x,y\in \sZ$ and take $z\in \partial \sZ$ such that $y=sx+(1-s)z$, where  $s=t_y(x)$. Then
\begin{align}
\label{eq:upbound1}
\|x-y\|_B&=\|x-sx-(1-s)z\|_B \nonumber \\
&=(1-s)\|x-z\|_B\le 2(1-b(y))
\end{align}
constitutes the upper bound derived in \cite{haapasalo2014}.

\begin{prop}\label{prop:attained}
Let $y\in ri(\sZ)$ and let $x\in\sZ$. The following are equivalent.

\begin{enumerate}
\item[(i)] $\|y-x\|_B=2(1-b(y))$
\item[(ii)] $t_y(x)=b(y)$
\item[(iii)] There is some $f\in C^*$, with $\|f\|_e=1$ and $\<f,y\>=b(y)$, such that $\<f,x\>=1$.

\end{enumerate}
\end{prop}

\begin{proof}  Suppose (i) and let $y=sx+(1-s)z$ with $s=t_y(x)$. Then
\[
2(1-b(y))=\|x-y\|_B=(1-s)\|x-z\|_B.
\]
Since both $(1-s)\le 1-b(y)$ and $\|x-z\|_B\le 2$, the equality implies that $t_y(x)=s=b(y)$.

Suppose (ii), then $y=b(y)x+(1-b(y))z$ for some $z\in \partial \sZ$. There is
some nonzero $f\in C^*$ such that $\<f,z\>=0$ and we may clearly suppose that
$\|f\|_e=1$. By Proposition \ref{prop:evalbd1}, $b(y)\le \<f,y\>=b(y)\<f,x\>\le
b(y)$. Since $y$ is an interior point, $b(y)>0$, so that we must have
$\<f,y\>=b(y)$ and $\<f,x\>=1$.

Finally, suppose (iii), then using inequalities
(\ref{eq:base}),(\ref{eq:upbound1}),
\begin{align*}
2(1-b(y))&\ge \|x-y\|_B\ge 2\<e-f,y-x\>=2\<e-f,y\>\\&=2(1-b(y)).
\end{align*}

\end{proof}

We now resolve the conjecture of the tightness of the upper bound
(\ref{eq:upbound1}) by showing that it can be always saturated.

\begin{theorem}\label{thm:tightness}
For any $y\in \sZ$, there exists some $x_0\in ext(\sZ)$, such that
\[
\|y-x_0\|_B=\sup_{x\in B}\|y-x\|_B=2(1-b(y)).
\]
\end{theorem}

\begin{proof}
Note first that since $x\mapsto \|y-x\|_B$ is a convex function, the supremum
over $\sZ$ is attained at some $x_0\in ext(\sZ)$. It is therefore enough to
prove that  equality in (\ref{eq:upbound1}) holds for some $x\in \sZ$. If $y$ is
an interior point, then by Proposition \ref{prop:attained}, the equality is
attained for any $x$ such that $t_y(x)=b(y)$, and we know from the results in
\cite{haapasalo2014} that this is achieved in
$B$. If $y\in \partial \sZ$, then there exists some $f\in C^*$, $\|f\|_e=1$ such
that $\<f,y\>=0$ and since $e-f\in \partial C^*$, there is some $x\in \sZ$ such
that $\<e-f,x\>=0$. Then
\[
2\ge \|y-x\|_B\ge 2\<e-f,y-x\>=2=2(1-b(y)).
\]
\end{proof}

\section{Boundariness for quantum channels}

In Ref.~\cite{haapasalo2014} it was shown that the inequality
(\ref{eq:upbound1}) is saturated for states and observables, however, the case
of channels remained open. Theorem~\ref{thm:tightness} shows that this
saturation holds also in this remaining case. In particular, for any interior
point $Y\in\cal{Q}$, where $\cal{Q}$ is either the set of quantum states, or
channels, or observables, the identity holds
$$
||X-Y||_B=2(1-b(Y) )\,,
$$
for a suitable $X\in ext(\cal{Q})$. In what follows we will make a bit stronger
and surprising observation that $X$ needs to be a unitary channel. We will prove
a theorem indicating that unitary channels are somehow special from the
perspective of boundariness and minimum-error discrimination.

\begin{lemma}
\label{lem:optsfd}
Let $D$ be a positive operator on $\cH_d\otimes\cH_d$ and
define
\begin{equation}\label{eqn:def-S}
\mathcal{R} = \left\{\ket{y}\in\cH_d\otimes\cH_d:
{\rm tr}_1\ket{y}\bra{y}\leq \frac{1}{d}I\right\}.
\end{equation}
Denote by $\ket{y_D}\in\mathcal{R}$ a vector which maximizes the overlap with
$D$, i.e. $\bra{y_D} D \ket{y_D} = \max_{\ket{y} \in \mathcal{R}} \bra{y} D
\ket{y}$. Then $\ket{y_D}$ is a unit vector, hence it is maximally entangled.
\end{lemma}
\begin{proof}
Let us note that $\ket{y}\in\mathcal{R}$ is normalized to one if and only if
$\ket{y}$ is maximally entangled, i.e. ${\rm tr}_1\ket{y}\bra{y}=\frac{1}{d}I$.
Suppose $\ket{y_D}$ has the following Schmidt decomposition $\ket{y_D}=\sum_j
\sqrt{\mu_j} \ket{e_j}\ket{f_j}$ and assume that for
some $k$ we have $\mu_k < 1/d$, thus it is not normalized. Then
\begin{align}
\bra{y_D} D \ket{y_D}
=&\, \mu_k \bra{e_k f_k} D \ket{e_k f_k} +
\sum_{j,l\neq k} \sqrt{\mu_j \mu_l} \bra{e_j f_j} D \ket{e_l f_l} \nonumber \\
& + 2 \sqrt{\mu_k} \sum_{j\neq k} \sqrt{\mu_j} Re \bra{e_k f_k} D \ket{e_j f_j}\,. \nonumber
\end{align}

In what follows we will construct a vector from $\mathcal{R}$ which has a
greater overlap with $D$. First, we introduce vector $\ket{\tilde{e}_k}$ which
differs from $\ket{e_k}$ only by a sign
\begin{equation}
\ket{\tilde{e}_k} = \mathrm{sgn}_+\left(\sum_{j\neq k} \sqrt{\mu_j} Re \bra{e_k f_k} D \ket{e_j f_j}\right) \ket{e_k},
\end{equation}
where $\mathrm{sgn}_+(x)$ equals to $1$ for non-negative $x$ and $-1$ for
negative $x$. Using this vector we write
\begin{equation}
\begin{split}
&
\mu_k \bra{e_k f_k} D \ket{e_k f_k} +
2 \sqrt{\mu_k} \sum_{j=1, j\neq k}^d \sqrt{\mu_j} Re \bra{e_k f_k} D \ket{e_j f_j}\\
&\leq
\mu_k \bra{e_k f_k} D \ket{e_k f_k} +
2 \sqrt{\mu_k} \left| \sum_{j=1, j\neq k}^d \sqrt{\mu_j} Re \bra{e_k f_k} D \ket{e_j f_j} \right|\\
&=
\mu_k \bra{\tilde{e}_k f_k} D \ket{\tilde{e}_k f_k} +
2 \sqrt{\mu_k} \sum_{j=1, j\neq k}^d \sqrt{\mu_j} Re \bra{\tilde{e}_k f_k} D \ket{e_j f_j}.
\end{split}
\end{equation}
In the last line above, $\mu_k$ is multiplied by strictly positive factor ($D$
is a positive matrix) and $\sqrt{\mu_k}$ is multiplied by a non-negative factor,
so we will (strictly) increase the value of the products if we replace $\mu_k$
with $\frac1d$. Finally we obtain
\begin{equation}
\begin{split}
\bra{y} D \ket{y} < \bra{\tilde{y}} D \ket{\tilde{y}},
\end{split}
\end{equation}
for $\ket{\tilde{y}} = \sum_{i=1,i\neq k}^d \sqrt{\mu_i} \ket{e_i f_i} +
\sqrt{\frac1d} \ket{\tilde{e}_k f_k}$. Since $\ket{\tilde{y}} \in \mathcal{R}$,
we obtained a contradiction.
\end{proof}

\begin{theorem}
\label{thbfore}
Suppose $\cF$ is an interior element of the set of channels $\cal{Q}$. Then
\begin{equation}
\label{eq:bforf}
b(\cF) = \left[{\max_{\cU} \lambda_1(J^{-1}_\cF J_\cU)}\right]^{-1}
 = \frac{d}{\max_{U} \bb U |J^{-1}_\cF |U\kk} \,,
\end{equation}
where the optimization runs over all unitary channels $\cU:\rho\mapsto U\rho
U^\dagger$ and $|U\kk=(U\otimes I) \sum_j \ket{jj}$. Moreover, if
$\cF=b(\cF)\,\mathcal{E} + (1-b(\cF))\, \mathcal{G}$ for some $\mathcal{E}\in
\cal{Q}$, $\mathcal G\in \partial \cal{Q}$, then $\mathcal{E}$ must be a unitary
channel.
\end{theorem}
\begin{proof}
Let us denote by $J_\cE,  J_\cF$ Choi-Jamiolkowski operators for channels $\cE$
and $\cF$, respectively. We assume $\cF$ is an interior element, thus, $J_{\cF}$
is invertible. Then $t_\cF(\cE)=\sup\{0\leq t<1, J_\cF-tJ_\cE\geq 0\}$. It
follows that for all $\ket{x}$, $\bra{x}J_{\cF}\ket{x} \geq  t \bra{x} J_{\cE}
\ket{x}$. Setting $\ket{y} = \sqrt{J_{\cF}} \ket{x}$ we obtain
\begin{align}
\label{eq:lbont}
\frac{1}{t} \geq \frac{\bra{y}\sqrt{J_{\cF}}^{-1} J_{\cE}
\sqrt{J_{\cF}}^{-1} \ket{y}}{\scalar{y}{y}}.
\end{align}
The maximum value of the right hand side equals
$\lambda_1(\sqrt{J_{\cF}}^{-1} J_{\cE} \sqrt{J_{\cF}}^{-1})
= \lambda_1(J_{\cF}^{-1} J_{\cE})=\lambda_1(\sqrt{J_{\cE}}J_{\cF}^{-1}\sqrt{J_{\cE}})$,
where $\lambda_1(X)$ denotes the maximal eigenvalue of $X$.
In conclusion, $t_\cF(\cE)=1/\lambda_1(J^{-1}_\cF J_\cE)$ and
\begin{equation}\label{eqn:formula-for-b}
b(\cF) = \inf_{\cE} t_\cF(\cE) = \left[{\max_{\cE} \lambda_1(J^{-1}_\cF J_\cE)}\right]^{-1}\,,
\end{equation}
where the optimization runs over all channels.

For any Choi-Jamio\l{}kowski state $J_\cE$ and an arbitrary unit vector
$\ket{x}\in\cH_d\otimes\cH_d$ we have $\sqrt{J_\cE} \ketbra{x}{x}\sqrt{J_\cE}  \leq  J_\cE$.
The complete positivity of partial trace implies
${\rm tr}_1 \left( J_\cE - \sqrt{J_\cE} \ketbra{x}{x}\sqrt{J_\cE} \right) \geq 0$,
and since ${\rm tr}_1 J_\cE  = \frac{1}{d} I$ it follows
$${\rm tr}_1 \sqrt{J_\cE} \ketbra{x}{x} \sqrt{J_\cE} \leq \frac{1}{d} I\,.$$
In other words, $\sqrt{J_\cE} \ket{x}\in\mathcal{R}$ defined in Lemma~\ref{lem:optsfd}.
Consequently, $\lambda_1(J_\cF^{-1} J_\cE)=
\max_{\ket{x}}\bra{x}\sqrt{J_\cE} J_\cF^{-1}\sqrt{J_\cE}\ket{x}
\leq\max_{\ket{y}\in\mathcal{R}} \bra{y}J_\cF^{-1}\ket{y}$
for every channel $\cE$
and using Eq. (\ref{eqn:formula-for-b}) we obtain
\begin{align}
\label{eq:lbforb}
b(\cF) = \left[\max_{\cE,\ket{x}}\bra{x}\sqrt{J_\cE} J_\cF^{-1}\sqrt{J_\cE}\ket{x}\right]^{-1}\geq \left[\max_{\ket{y}\in\mathcal{R}} \bra{y}J_\cF^{-1}\ket{y}\right]^{-1}.
\end{align}
Since $J_\cF^{-1}$ is a positive operator Lemma \ref{lem:optsfd} implies that
the maximum over $\ket{y}$ is achieved only by unit (hence maximally entangled)
vectors. For every such vector $\ket{y_\cF}$ there exists a unitary matrix $U$
such that $\ket{y_\cF}=\frac{1}{\sqrt{d}}\sum_j U\ket{j}\otimes \ket{j}$.
Moreover, choice of $\ket{x}=\ket{y_\cF}$, $\cE=\cU$, where
$J_{\cU}=\ket{y_\cF}\bra{y_\cF}$ proves that the lower bound (\ref{eq:lbforb})
is tight. Finally, the achievability of maximum on the right hand side of
Eq.(\ref{eq:lbforb}) requires by Lemma \ref{lem:optsfd} that the norm of
$\sqrt{J_\cE}\ket{x}$ is one, which in turn implies that $\cE$ is a unitary
channel. Otherwise $t_\cF(\cE)> b(\cF)$ (see Eq. (\ref{eqn:formula-for-b})) and
decompositions of the form $\cF=b(\cF)\cE+(1-b(\cF))\mathcal{G}$
($\mathcal{G}\in\partial\cal{Q}$) can not exist.
\end{proof}

\begin{corollary}
Suppose $\cF$ is an interior element of the set of channels. Then there exist a
unitary channel $\cU$ such that $||\cF-\cU||_{B}=2(1-b(\cF))$. Moreover, if
$\cE\in \cal{Q}$ is not a unitary channel, then $\|\cF-\cE\|_B<2(1-(b(\cF))$.
\end{corollary}

\begin{proof}
Combining Proposition \ref{prop:attained} and Theorem \ref{thbfore} we conclude
that  the equality $||\cF-\cU||_{B}=2(1-b(\cF))$ holds precisely for  unitary
channels $\cU$ such that $\frac{b(\cF)}{d}=\bb U |J^{-1}_\cF |U\kk ^{-1}$
\end{proof}

In what follows we will explicitly evaluate the boundariness formula
determined in Eq.~(\ref{eq:bforf}) for the families of
qubit and erasure channels (on arbitrary dimensional system).

%%%%%%%%%%%%%%%%%%%%%%%%%%%%%%%%%%%%%%%%%%%%%%%%%%%%%%%%%%%%%%%%%%%%%%%%%%%%%%%
\subsection{Qubit channels}
%%%%%%%%%%%%%%%%%%%%%%%%%%%%%%%%%%%%%%%%%%%%%%%%%%%%%%%%%%%%%%%%%%%%%%%%%%%%%%%
\begin{theorem}
\label{thbforqubit}
Suppose $\cF$ is an interior element of the set of qubit channels.
Then
\begin{equation}
\label{eq:bforfqubit}
b(\cF) =\frac{2}{
\lambda_1
\left(
W^{\dagger} J_{\cF}^{-1} W + (W^{\dagger} J_{\cF}^{-1} W)^T
\right)
}\, ,
\end{equation}
where $W$ is a unitary matrix (called sometimes a Magic
Basis)~\cite{hill1997entanglement}
\begin{equation}
W = \frac{1}{\sqrt{2}}
\left(
\begin{smallmatrix}
0 & 0 & 1 & \ii \\
-1 & \ii & 0 & 0 \\
1 & \ii & 0 & 0 \\
0 & 0 & 1 & -\ii
\end{smallmatrix}
\right).
\end{equation}
\end{theorem}
\begin{proof}
For any qubit channel $\cF$ with Choi-Jamio\l{}kowski state $J_{\cF}$, boundariness $b(\cF)$
is given by (see Eq.~\eqref{eq:bforf})
\begin{equation}
\label{eq:numrange}
b(\cF) =\frac{1}{\max_{\psi\in\cS_{ME}} \bra{\psi}J^{-1}_\cF \ket{\psi} } \equiv \frac{1}{r^{\mathrm{ent}}\left(J^{-1}_{\cF}\right)},
\end{equation}
where $\cS_{ME}=\left\{\ket{\psi}\in \cH_d\otimes\cH_d\,|\,{\rm
tr}_1\ket{\psi}\bra{\psi}=\frac{1}{d}\,I \right\}$ and $r^{\mathrm{ent}}(A)$ is
a maximally entangled numerical radius for
matrix $A$.
We know from the literature~\cite{dunkl2014real}, that maximally entangled
numerical range for $4 \times 4$ matrix  $A$ is equal to real numerical range of
matrix $W^{\dagger} A W$.
From the above we note, that
\begin{equation}
r^{\mathrm{ent}}(J_{\cF}^{-1}) =
\lambda_1
\left(
\frac{W^{\dagger} J_{\cF}^{-1} W + (W^{\dagger} J_{\cF}^{-1} W)^T}{2}
\right),
\end{equation}
which together with Eq. (\ref{eq:numrange}) finishes the proof.
\end{proof}

In the case of qubit channel $\cF$ we can specify, the unitary channel $\cU$,
for which  $||\cF-\cU||_{B}=2(1-b(\cF))$. It follows from the reasoning above,
that unitary matrix $U$, which defines the channel, can be written as
\begin{equation}
| U \kk = \sqrt{2} W \ket{v}.
\end{equation}
Vector $\ket{v}$ above is the leading eigenvector of real symmetric matrix
$W^{\dagger} J^{-1}_\cF W + (W^{\dagger} J^{-1}_\cF W)^T$.

%%%%%%%%%%%%%%%%%%%%%%%%%%%%%%%%%%%%%%%%%%%%%%%%%%%%%%%%%%%%%%%%%%%%%%%%%%%%%%%
\subsection{Erasure channels}\label{sec:erasure-example}
%%%%%%%%%%%%%%%%%%%%%%%%%%%%%%%%%%%%%%%%%%%%%%%%%%%%%%%%%%%%%%%%%%%%%%%%%%%%%%%
Erasure channels transform any input state $\rho$ onto a fixed output state $\cF_\sigma(\rho) = \sigma$. For such channel $\cF_\sigma$ the Choi-Jamio\l{}kowski state reads
\begin{equation}
J_{\cF_\sigma} = \frac{1}{d} \sigma \otimes I.
\end{equation}

\begin{prop}
Boundariness of erasure channel $\cF_\sigma$, which maps everything to a fixed
interior point  $\sigma$ in the set of states $\cS(\cH_d)$, is given by
\begin{equation}
b(\cF_{\sigma})
=\frac{1}{\tr{\sigma^{-1}}}.
\end{equation}
\end{prop}

\begin{proof}
Since $\sigma$ is an interior element of the set of states,
$J^{-1}_{\cF_\sigma} = d\, \sigma^{-1} \otimes I$ is well defined.
Using theorem \ref{thbfore} we obtain
\begin{equation}
b(\cF_\sigma)=\frac{1}{\max_{U} \sum_{j,k} \bra{jj} (U^\dagger\sigma^{-1} U)\otimes I \ket{kk} }=\frac{1}{\tr{\sigma^{-1}}}, \nonumber
\end{equation}
where we used $U\,U^\dagger=I$ and the cyclic invariance of the trace.
\end{proof}
Let us note that in the special case of a qubit erasure channel $\cF_\sigma$ with $\sigma=p \ket{0}\bra{0}+(1-p)\ket{1}\bra{1}$ we find $b(\cF_\sigma)=p(1-p)$ in accordance with the results of \cite{haapasalo2014}.

%%%%%%%%%%%%%%%%%%%%%%%%%%%%%%%%%%%%%%%%%%%%%%%%%%%%%%%%%%%5
\section{Boundariness under composition}
%%%%%%%%%%%%%%%%%%%%%%%%%%%%%%%%%%%%%%%%%%%%%%%%%%%%%%%%%%%5
Suppose $\mathcal{E}, \mathcal{F}$ are channels on systems
described in Hilbert spaces $\cH_s$,$\cH_d$, respectively. Denote by
$b(\mathcal{E}), b(\cF)$ the values of their boundariness.
In this section we address the question of the relation between
the boundariness of channel composition, $b(\cE\otimes\cF)$, and
the boundariness for individual channels.

\begin{prop}\label{prop:tensor_general}
For channels the boundariness is sub-multiplicative, i.e.
$b(\mathcal{E}\otimes\mathcal{F})\leq b(\mathcal{E})b(\mathcal{F})$.
\end{prop}

\begin{proof}
Let us consider some decomposition of channels $\cE,\cF$ into boundary elements with the weight equal to their boundariness.
\begin{align}
J_\cE&=b(\cE)J_{\cE+}+[1-b(\cE)]J_{\cE-} \nonumber \\
J_\cF&=b(\cF)J_{\cF+}+[1-b(\cF)]J_{\cF-} \nonumber
\end{align}
This allows us to write:
\begin{align}
\label{eq:decomptp1}
J_\cE\otimes J_\cF&=b(\cE)\, b(\cF)\,J_{\cE+}\otimes J_{\cF+}+[1-b(\cE)\,b(\cF)\,]\,J_{\mathcal{T}}, \nonumber \\
\end{align}
where
\begin{align}
J_{\mathcal{T}}=& [1-b(\cE)b(\cF)]^{-1}\bigl( b(\cE)[1-b(\cF)]\,J_{\cE+}\otimes J_{\cF-} \nonumber \\
&+[1-b(\cE)]\,b(\cF)\, J_{\cE-}\otimes J_{\cF+} \nonumber \\
& +[1-b(\cE)]\,[1-b(\cF)]\, J_{\cE-}\otimes J_{\cF-}\bigr)
\end{align}
is a Choi-Jamiolkowski state of a channel.
Let us remind that a channel is on the boundary of the set of channels  if and
only if its Choi-Jamiolkowski state has non empty kernel (see e.g.
\cite{haapasalo2014}). It is easy to realize that if $\cE_+$ and  $\cF_+$ are
boundary elements of the respective sets of channels, $\cE_+\otimes \cF_+$ lies
on the boundary as well. Similarly, taking vectors $\ket{\varphi}, \ket{\psi}$
from the kernel of $J_{\cE_-}$, $J_{\cF_-}$, respectively, we can immediately
see that $\ket{\varphi}\otimes\ket{\psi}$ belongs to the kernel of
$J_\mathcal{T}$. This shows that Eq. (\ref{eq:decomptp1}) provides a valid
convex decomposition of a channel $\mathcal{E}\otimes\mathcal{F}$ into two
boundary elements  and we conclude
$t_{\mathcal{E}\otimes\mathcal{F}}(\mathcal{E}_+\otimes\mathcal{F}_+)=b(\mathcal{E})b(\mathcal{F})$.
Due to definition of boundariness from Eq. (\ref{eq:defbdorig}) we obtain the
upper bound from the proposition.
\end{proof}

\begin{prop}\label{prop:tensor}
For states and observables the boundariness is multiplicative, i.e.
$b(x\otimes y)=b(x)b(y)$, where $x,y$ stands for any pair of states,
or observables.
\end{prop}
\begin{proof}
The equality in Proposition \ref{prop:tensor} is fulfilled, because for states
and observables the boundariness is given by the smallest eigenvalue and
eigenvalues of the tensor products are products of the eigenvalues.
\end{proof}

We have numerical evidence suggesting that equality holds also in the case of
channels, but we have no proof of such conjecture. Using Eq. (\ref{eq:bforf}),
this is equivalent
to equality of $\max_{\xi}\bra{\xi}J^{-1}_{\cE}\otimes J^{-1}_{\cF}\ket{\xi}$
and $\max_{\chi}\bra{\chi}J^{-1}_{\cE}\ket{\chi} \max_\omega
\bra{\omega}J^{-1}_{\cF}\ket{\omega},$
where $\xi, \chi, \omega$ are maximally entangled states on the corresponding
systems.

Below we prove this equality for case of qubit channels when one of the channels
is the ''maximally mixed'' channel $\cF$, hence, for this pair of channels the
boundariness is multiplicative.

\begin{prop}\label{prop:product}
Let $\cE$ be an arbitrary qubit channel and let $\cF$ be the erasure channel
mapping any input to $\frac{1}{d} I$. Then $b(\cE\otimes \cF)=b(\cE)b(\cF)$.
\end{prop}

\begin{proof}
By Proposition \ref{prop:tensor}, $b(\cE\otimes \cF)\le b(\cE)b(\cF)$, so that
we have to show the opposite inequality.  Let $\cE: \cB(\cH_A)\to \cB(\cH_B)$
and $\cF: \cB(\cH_{A'})\to \cB(\cH_{B'})$, where $\cH_A$, $\cH_B$ denote copies
of $\cH_2$,  and $\cH_{A'}$, $\cH_{B'}$ denote copies of $\cH_d$.  Since
$J_{\cF}^{-1}=d^2 I_{B'A'}$ then by Theorem \ref{thbfore} we want to prove the
following inequality
\[
\max_{V\in \cU(\cH_{BB'})} \bb V| J^{-1}_{\cE}\otimes I_{B'A'}|V\kk\le d\max_{U\in \cU(\cH)} \bb U |J^{-1}_\cE |U\kk\,.
\]
For $V\in \cU(\cH_{BB'})$, let $X_V=\mathrm{tr}_{B'A'}|V\kk\bb V|$. Then $X_V$ is a positive operator on $\cH_{BA}$ and we have
\[
\mathrm{tr}_{B}X_V=\mathrm{tr}_{A'}\mathrm{tr}_{BB'} |V\kk\bb V|=d I_A.
\]
Similarly, $\mathrm{tr}_AX_V=d I_B$. It follows that $\frac{1}{2d} X_V$ is the Choi-Jamiolkowski matrix of a unital qubit channel. As it is well known, any such channel is a random unitary channel, so that there are some unitaries $U_i\in \cU(\cH_2)$ and probabilities $p_i$ such that $X_V=d\sum_ip_i |U_i\kk\bb U_i|$. It follows that
\[
\bb V| J^{-1}_{\cE}\otimes I_{B'A'}|V\kk=\tr{J_{\cE}^{-1}X_V} \le d\max_{U\in \cU(\cH)} \bb U|J_\cE^{-1}|U\kk.
\]
\end{proof}

\section{Maximal value of boundariness}
By definition, boundariness takes values between zero and one half, but all
values in this interval are not necessarily attained. A simple example is the
triangle (see Fig.~\ref{fig:subfigure1}),
where one third is the maximal value. In this section we will
investigate, what is the highest achievable value of boundariness in quantum
convex sets, and which are the points achieving it. In fact, we will see that
such point is unique and coincides with so-called maximally mixed element.

As for the other questions addressed in this paper, it is straightforward to
evaluate the maximal value for states and measurements, but the case of channels
is more involved.

\begin{prop}
The maximal value of boundariness for quantum convex sets is given
as follows:
\begin{itemize}
\item{\emph{States:}} $b_{\max}^s=1/d$ achieved for completely mixed state $\varrho=\frac{1}{d}I$.
\item{\emph{Observables:}} $b_{\max}^o=1/n$ achieved for $n$-outcome (uniformly)
trivial observable $\{E_j=\frac{1}{n}I\}_{j=1}^n$.
\item{\emph{Channels:}} $b_{\max}^c=1/d^{2}$ achieved for completely depolarizing channel mapping all states into completely mixed state $\frac{1}{d}I$.
\end{itemize}
\end{prop}
\begin{proof}
For states and measurements \cite{haapasalo2014} the highest boundariness means
highest value of the lowest eigenvalue, which leads to maximally mixed state
$\rho=\frac{1}{d}I$ and (uniform) trivial observable
$\{E_i=\frac{1}{N}I\}_{i=1}^N$, respectively. The case of channels is more
subtle. From the formula (\ref{eq:bforf}) giving the boundariness of a channel
it is clear that we search for a channel $\cF$ such that ${\max_{U} \bb U
|J^{-1}_{\cF}| U \kk }$ is minimized. We construct a simple lower bound using an
orthonormal basis $\{\ket{v_i}\}_{i=1}^{d^2}$ of maximally entangled states.
\begin{align}
\label{eq:lbmaxu}
\tr{J^{-1}_{\cF}}=\sum_{i=1}^{d^2} \bra{v_i} J^{-1}_{\cF}\ket{v_i}
\leq d \max_{U} \bb U |J^{-1}_{\cF}| U \kk.
\end{align}
Such a basis $\{\ket{v_{pq}}=Z^p W^q\otimes I \frac{1}{\sqrt{d}}\sum_j
\ket{jj}\}$ can be constructed by Shift and multiply unitary operators $Z=\sum_j
\ket{j\oplus 1}\bra{j}$, $W=\sum_j \omega^j \ket{j}\bra{j}$, where
$\omega=e^{\frac{2\pi i}{d}}$.
On the other hand from spectral decomposition $J_{\cF}=\sum_i \lambda_i
\ket{a_i}\bra{a_i}$, where $\sum_i \lambda_i=1$, we have
$\tr{J^{-1}_{\cF}}=\sum_i \frac{1}{\lambda_i}\geq d^4$. Combining this with Eq.
(\ref{eq:lbmaxu}) we get $d^3\leq \max_{U} \bb U |J^{-1}_{\cF}| U \kk$.
Inserting this into Eq. (\ref{eq:bforf}) we finally obtain $b(\cF)\leq
\frac{1}{d^2}$. It is easy to see that the inequalities can be made tight only
by a single channel, which maps everything to a complete mixture.
\end{proof}

%%%%%%%%%%%%%%%%%%%%%%%%%%%%%%%%%%%%%%%%%%%%%%%%%%%%%%%%%%%%%%%%%%%%%%%%
\section{Summary}
%%%%%%%%%%%%%%%%%%%%%%%%%%%%%%%%%%%%%%%%%%%%%%%%%%%%%%%%%%%%%%%%%%%%%%%%
This paper completes and extends the previous work~\cite{haapasalo2014}
in which the concept of boundariness was introduced. We proved that
for compact convex sets
evaluation of boundariness of $y$
coincides with the question of the best distinguishable element from $y$, i.e.
$$
2(1-b(y))=\max_x ||x-y||\,,
$$
where $||\cdot||$ denotes the so-called base norm (being
trace-norm for states, completely bounded norm -- also known as the diamond norm for channels
and observables). This identity was formulated in Ref.\cite{haapasalo2014}
as an open conjecture for case of quantum channels and is confirmed by our
results presented in this paper. In fact, we have discovered that
the optimum is attained only for unitary channels. This surprising
result provides quite unexpected operational characterization of
unitary channels and exhibits their specific role among boundary elements
and in minimum error discrimination questions.
The unique role of unitary channels is noticeable also in the explicit formula that we derived for the evaluation of boundariness of channels.
In the current paper we investigated only quantum channels mapping between Hilbert spaces of the same dimension. The results can be easily generalized for the case when the input has smaller dimension than the output. The role of unitary channels will be played by isometries. The opposite relation of the input/output dimensions seems to be much more complicated and is left for future research.
Further we investigated how the boundariness behaves under the tensor
product. We have shown that boundariness is a multiplicative quantity
for states and observables, however, for channels we proved only the sub-multiplicativity
$$
b(\cE\otimes\cF)\leq b(\cE)b(\cF)\,.
$$
However, our numerical analysis suggests that the boundariness is multiplicative
also for case of channels.

Exploiting the relation between the boundariness and the discrimination,
the multiplicativity implies that the most distinguishable element from
$x\otimes y$ is still a factorized element $x_0\otimes y_0$, where
$x_0,y_0$ stands for the most distinguishable elements from $x,y$,
respectively. For channels this would mean that factorized unitaries
are the most distant ones for all factorized channels. However, whether
this is the case is left open.

In the remaining part of the paper we evaluated explicitly the maximal
value of boundariness. We found that this maximum is achieved
for intuitively the maximally mixed elements, i.e. for completely mixed state,
uniformly trivial observables and channel contracting state space to the
completely mixed state.  In particular, for $d$-dimensional
quantum systems we found for states $b_{\max}^s=1/d$, for
observables $b_{\max}^o  = 1/n$ is independent on the dimension (only
the number of outcomes $n$ matters), and for channels $b^c_{\max}=1/d^2$. Let us
stress that these numbers also determine the optimal values of error
probability for related discrimination problems.

%%%%%%%%%%%%%%%%%%%%%%%%%%%%%%%%%%%%%%%%%%%%%%%%%%%%%%%%%%%%%%%%%%%%%%%%
%%%%%%%%%%%%%%%%%%%%%%%%%%%%%%%%%%%%%%%%%%%%%%%%%%%%%%%%%%%%%%%%%%%%%%%%
%%%%%%%                        ACKNOWLEDGMENT                     %%%%%%
%%%%%%%%%%%%%%%%%%%%%%%%%%%%%%%%%%%%%%%%%%%%%%%%%%%%%%%%%%%%%%%%%%%%%%%%
%%%%%%%%%%%%%%%%%%%%%%%%%%%%%%%%%%%%%%%%%%%%%%%%%%%%%%%%%%%%%%%%%%%%%%%%
\acknowledgments
We thank to Errka Hapaasalo for discussions and  workshop ceqip.eu
for initiating this work. This work was supported by project
VEGA 2/0125/13 (QUICOST).
Z.P. acknowledges a support from the Polish National Science Centre
through grant number DEC-2011/03/D/ST6/00413.
A.J. acknowledges support by Research and Development Support Agency
under the contract No. APVV-0178-11 and  VEGA 2/0059/12.
M.S. acknowledges support by the Operational Program Education for
Competitiveness—-European Social Fund (Project No. CZ.1.07/2.3.00/30.0004)
of the Ministry of Education, Youth and Sports of the Czech
Republic.
M.Z. acknowledges the support of GA\v CR project P202/12/1142
and COST Action MP1006.

%%%%%%%%%%%%%%%%%%%%%%%%%%%%%%%%%%%%%%%%%%%%%%%%%%%%%%%%%%%%%%%%%%%%%%%%%%
%%%%%%%%%%%%%%%%%%%%%%%%%%%%%%%%%%%%%%%%%%%%%%%%%%%%%%%%%%%%%%%%%%%%%%%%%%
%%%%%                      BIBLIOGRAPHY                             %%%%%%
%%%%%%%%%%%%%%%%%%%%%%%%%%%%%%%%%%%%%%%%%%%%%%%%%%%%%%%%%%%%%%%%%%%%%%%%%%
%%%%%%%%%%%%%%%%%%%%%%%%%%%%%%%%%%%%%%%%%%%%%%%%%%%%%%%%%%%%%%%%%%%%%%%%%%

%%%%%%%%%%%%%%%%%%%%%%%%%%%%%%%%%%%%%%%%%%%%%%%%%%%%%%%%%%%%%%%%%%%%%%%
%%%%%%%%%%%%%%%%%%%%%%%%%%%%%%%%%%%%%%%%%%%%%%%%%%%%%%%%%%%%%%%%%%%%%%%
%%%%                     END OF DOCUMENT                           %%%%
%%%%%%%%%%%%%%%%%%%%%%%%%%%%%%%%%%%%%%%%%%%%%%%%%%%%%%%%%%%%%%%%%%%%%%%
%%%%%%%%%%%%%%%%%%%%%%%%%%%%%%%%%%%%%%%%%%%%%%%%%%%%%%%%%%%%%%%%%%%%%%%

\end{document}